\documentclass[11pt,a4paper]{article}
\usepackage[margin=1in]{geometry}
\usepackage{xspace}
\usepackage{amsmath,amssymb,amsthm}
\usepackage{graphics}
\usepackage{graphicx}
\usepackage{hyperref}
\usepackage{tikz}
\usepackage{paralist}
\usepackage{wrapfig}
\usepackage{bibunits}

\usepackage{subfigure}
\usepackage{tabularx}
\usepackage{multirow}

\usepackage[algo2e, ruled, vlined, linesnumbered]{algorithm2e}
\DontPrintSemicolon

\newtheorem{theorem}{Theorem}
\newtheorem{lemma}{Lemma}
\newtheorem{corollary}{Corollary}

\defaultbibliography{fastSimpleDynCutTree_tr}
\defaultbibliographystyle{plain}

\def\comment#1{}

\def\withcomments{
  \newcounter{mycommentcounter}
  \def\comment##1{\refstepcounter{mycommentcounter}%
   \ifhmode%
    \unskip%
    {\dimen1=\baselineskip \divide\dimen1 by 2 %
      \raise\dimen1\llap{\tiny -\themycommentcounter-}}\fi%
    \marginpar{\renewcommand{\baselinestretch}{0.8}%
      \footnotesize [\themycommentcounter]: \raggedright ##1}}
  }
\withcomments

\newcommand{\algo}[1]{{\normalfont\scshape #1}}   
\newcommand{\dcup}{\ensuremath{\mathaccent\cdot\cup}}
\newcommand{\Gp}{\ensuremath{G^\oplus}}
\newcommand{\Gm}{\ensuremath{G^\ominus}}
\newcommand{\cmm}{\ensuremath{c_\ominus}}

\SetKwComment{myco}{// }{\DontPrintSemicolon}
\graphicspath{{.}{./fig/}}
\title{Dynamic Gomory-Hu Tree Construction---fast and simple\thanks{This
    work was partially supported by the DFG under grant WA 654/15-2
    and by the Concept for the Future of Karlsruhe Institute of
    Technology within the German Excellence Initiative.  }} \author{
  Tanja Hartmann and Dorothea Wagner } \date{ Department of
  Informatics, Karlsruhe Institute of Technology (KIT)\\
  \texttt{\{t.hartmann,dorothea.wagner\}@kit.edu} }
\begin{document}

\begin{bibunit}
\setcounter{totalnumber}{8}
\setcounter{topnumber}{8}
\setcounter{bottomnumber}{8}
\renewcommand{\textfraction}{0.001}
\renewcommand{\topfraction}{0.999}
\renewcommand{\bottomfraction}{0.999}
\maketitle
\begin{abstract}
  A cut tree (or Gomory-Hu tree) of an undirected weighted graph
  $G=(V,E)$ encodes a minimum $s$-$t$-cut for each vertex pair
  $\{s,t\} \subseteq V$ and can be iteratively constructed by $n-1$
  maximum flow computations.  They solve the multiterminal network
  flow problem, which asks for the all-pairs maximum flow values in a
  network and at the same time they represent $n-1$ non-crossing,
  linearly independent cuts that constitute a minimum cut basis of
  $G$.
Hence, cut trees are resident in at least two fundamental fields of network analysis and graph theory, which emphasizes their importance for many applications.
In this work we present a fully-dynamic algorithm that efficiently maintains a cut tree for a changing graph. The algorithm is easy to implement and has a high potential for saving cut computations under the assumption that a local change in the underlying graph does rarely affect the global cut structure.
We document the good practicability of our approach in a brief experiment on real world data.
\end{abstract}
\section{Introduction}\label{sec:Intro}
A \emph{cut tree} is a weighted tree $T(G) = (V, E_T,c_T)$ on the vertices of an undirected (weighted) graph $G = (V,E,c)$ (with edges not necessarily in $G$) such that each $\{u,v\} \in E_T$ induces a minimum $u$-$v$-cut in $G$ (by decomposing $T(G)$ into two connected components) and such that $c_T(\{u,v\})$ is equal to the cost of the induced cut.
The cuts induced by $T(G)$ are non-crossing and for each $\{x,y\} \subseteq {V}$ each cheapest edge on the path $\pi(x,y)$ between $x$ and $y$ in $T(G)$ corresponds to  a minimum $x$-$y$-cut in $G$. If~$G$ is disconnected,  $T(G)$ contains edges of cost $0$ between connected components.

Cut trees were first introduced by Gomory and Hu~\cite{gh-mtnf-61} in 1961 in the field of multiterminal network flow analysis. Shortly afterwards, in 1964, Elmaghraby~\cite{e-samfn-64} already studied
how the values of multiterminal flows change if the capacity of an edge in the network varies. Elmaghraby established the \emph{sensitivity analysis of multiterminal flow networks}, which asks for the all-pairs maximum flow values (or all-pairs minimum cut values) in a network considering any possible capacity of the varying edge. According to Barth et al.~\cite{bbdf-rm-06} this can be answered by constructing two cut trees.
In contrast, the \emph{parametric maximum flow problem} 
considers a flow network with only two terminals $s$ and $t$ and with several parametric edge capacities. The goal is to give a maximum $s$-$t$-flow (or minimum $s$-$t$-cut) regarding all possible capacities of the parametric edges.
Parametric maximum flows were studied, e.g., by Gallo et al.~\cite{ggt-fpmfaa-89} and Scutell\`a~\cite{s-a-06}. 

However, in many applications we are 
neither interested in \emph{all-pairs} values nor in one minimum $s$-$t$-cut 
regarding \emph{all possible} changes of varying edges.
Instead we face a concrete change on a concrete edge and need all-pairs minimum cuts
regarding this single change. This is answered by \emph{dynamic cut trees}, which thus
bridge the two sides of sensitivity analysis 
and parametric maximum flows. 
%
\paragraph{Contribution and Outline.}
In this work we develop the first algorithm that efficiently and dynamically maintains a cut tree for a changing graph allowing arbitrary atomic changes. 
To the best of our knowledge no fully-dynamic approach for updating cut trees exists.
Coming from sensitivity analysis, Barth et al.~\cite{bbdf-rm-06} state 
that after the capacity of an edge has increased the path in $T(G)$ between the vertices that define the changing edge in $G$ is the only part of a given cut tree that needs to be recomputed, which is rather obvious. Besides they stress the difficulty for the case of decreasing edge capacities. 

In our work we formulate a general condition for the (re)use of given cuts in an (iterative) cut tree construction, which directly implies the result of Barth et al.
We further solve the case of decreasing edge capacities showing by an experiment that this has a similar potential for saving cut computations like the case of increasing capacities.
In the spirit of Gusfield~\cite{g-vsmap-90}, who simplified the pioneering cut tree algorithm of Gomory and Hu~\cite{gh-mtnf-61}, we also allow the use of crossing cuts and give a representation of intermediate trees (during the iteration) that makes our approach very easy to implement.

We give our notational conventions and a first folklore insight in Sec.~\ref{sec:Prelim}.
In Sec.~\ref{sec:Theo} we revisit the static cut tree algorithm~\cite{gh-mtnf-61} and the key for its simplification~\cite{g-vsmap-90}, and construct a first intermediate cut tree by reusing cuts that obviously remain valid after a change in $G$. We also state several lemmas that imply techniques to find further reusable cuts in this section. Our update approach is described in Sec.~\ref{sec:DynCutTreeAlgo}. In Sec.~\ref{sec:Perform} we finally discuss the performance of our algorithm based on a brief experiment.
%
\paragraph{Preliminaries and Notation.}\label{sec:Prelim}
In this work we consider an undirected, weighted graph $G = (V,E,c)$ with vertices $V$, edges $E$ and a positive edge cost function~$c$, writing $c(u,v)$ as a shorthand for $c(\{u,v\})$ with 
$\{u,v\} \in E$. 
We reserve the term \emph{node} 
for compound vertices of abstracted graphs, which may contain several basic vertices of a concrete graph; however, we identify singleton nodes with the contained vertex without further notice.
\emph{Contracting} a set $N \subseteq V$ in $G$ means replacing $N$ by a single node, and leaving this node adjacent to all former adjacencies $u$ of vertices of $N$, with an edge cost equal to the sum of all former edges between $N$ and $u$.
Analogously we contract a set $M \subseteq E$ or a subgraph of~$G$ by contracting the corresponding vertices.

A \emph{cut} in $G$ is a partition of $V$ into two \emph{cut sides} $S$ and $V\setminus S$. The cost $c(S,V\setminus S)$ of a cut is the sum of the costs of all edges \emph{crossing} the cut, i.e., edges $\{u,v\}$ with $u\in S$, $v \in V\setminus S$. For two disjoint sets $A,B\subseteq V$ we define the cost $c(A,B)$ analogously. Note that a cut is defined by the edges crossing it. Two cuts are \emph{non-crossing} if their cut sides are pairwise nested or disjoint. 
Two vertices $u,v \in V$ are \emph{separated} by a cut if they lie on different cut sides.
A minimum $u$-$v$-cut is a cut that separates $u$ and $v$ and is the cheapest cut among all cuts separating these vertices. We call a cut a \emph{minimum separating cut} if there exists an arbitrary vertex pair $\{u,v\}$ for which it is a minimum $u$-$v$-cut; $\{u,v\}$ is called a \emph{cut pair} of the minimum separating cut.
We further denote the \emph{connectivity} of $\{u,v\}\subseteq V$ by $\lambda(u,v)$, describing the cost of a minimum $u$-$v$-cut.

Since each edge in a tree $T(G)$ on the vertices of $G$ induces a unique cut in $G$,
we identify tree edges with corresponding cuts without further notice. This allows for saying that a vertex is \emph{incident} to a cut and an edge \emph{separates} a pair of vertices.
%
We consider the path $\pi(u,v)$ between $u$ and $v$ in $T(G)$ as the set of edges or the set of vertices on it, as convenient.
%

A change in $G$ either involves an edge $\{b,d\}$ or a vertex $b$. If the cost of $\{b,d\}$ in~$G$ descreases by $\Delta > 0$ or $\{b,d\}$ with $c(b,d) = \Delta>0$ is deleted,
the change yields~$\Gm$. Analogously, inserting $\{b,d\}$ or increasing the cost yields $\Gp$. 
We denote the cost function after a change by  $c^\ominus$ and $c^\oplus$, the connectivity by $\lambda^\ominus$ and $\lambda^\oplus$, respectively. We assume that only degree-$0$ vertices can be deleted from $G$.
Hence, inserting or deleting $b$ changes neither the cost function nor the connectivity.
We start with a fundamental insight on the reusability of cuts. Recall that $T(G) = (V, E_T, c_T)$ denotes a cut tree.
\newcommand{\lemmaPath}{
If $c(b,d)$ changes by $\Delta > 0$, then
	each $\{u,v\} \in E_T$ remains a minimum $u$-$v$-cut (i) in $\Gp$ with cost $\lambda(u,v)$ if $\{u,v\} \notin \pi(b,d)$, 
	(ii) in $\Gm$ with cost $\lambda(u,v)-\Delta$ if $\{u,v\} \in  \pi(b,d)$.
}
\begin{lemma}
\label{lem:path}
	\lemmaPath
\end{lemma}
\begin{proof}
	The edges on $\pi(b,d)$ are the only edges in $E_T$ that represent cuts
	separating $b$ and $d$. Thus, these edges represent the only cuts with 
	changing costs in $T(G)$. The costs of those edges change exactly by $\Delta$.
	If $c(\{b,d\})$ decreases, let $\{u,v\} \in E_T$ and observe that the connectivity 
	$\lambda(u,v)$ decreases by at most $\Delta$. Hence, $\{u,v\}$ is a 
	minimum $u$-$v$-cut in $\Gm$, since $c_T(u,v) = \lambda(u,v) - \Delta$.
	If the cost of $\{b,d\}$ increases, the cuts whose costs do not change obviously 
	remain minimum separating cuts in $\Gp$.
\end{proof}

\section{The Static Algorithm and Insights on Reusable Cuts}\label{sec:Theo}
\paragraph{The Static Algorithm.}
As a basis for our dynamic approach, we briefly revisit the static construction of a cut tree~\cite{gh-mtnf-61,g-vsmap-90}. 
This algorithm iteratively constructs $n-1$ non-crossing minimum separating cuts for $n-1$ vertex pairs, which we call \emph{step pairs}. These pairs are chosen arbitrarily from the set of pairs not separated by any of the cuts constructed so far. 
Algorithm~\ref{alg:gh} briefly describes the cut tree algorithm of Gomory and Hu.
\begin{algorithm2e}[h]
	\caption{\algo{Cut Tree}}\label{alg:gh}
	\KwIn{Graph $G=(V,E,c)$}
	\KwOut{Cut tree of $G$}
	Initialize tree $T_* := (V_*,E_*,c_*)$ with $V_* \gets \{V\}, E_* \gets \emptyset$ and $c_*$ empty\nllabel{alg:gh:init}\\ 
	\While({\myco*[f]{unfold all nodes}}){$\exists S \in V_*$ with $|S| > 1$}{
		$\{u,v\} \gets$ arbitrary pair from $\binom{S}{2}$ \nllabel{alg:gh:uv}\\
		\lForAll{$S_j$ adjacent to $S$ in $T_*$}{$N_j \gets$ subtree of $S$ in $T_*$ with $S_j \in N_j$} \nllabel{alg:gh:subtrees}\\
		$G_S = (V_S,E_S,c_S) \gets$ in $G$ contract each $N_j$ to $[N_j]$ \myco*[f]{contraction} \nllabel{alg:gh:contract}\\
 		$(U,V \setminus U) \gets$  min-$u$-$v$-cut in $G_S$, cost $\lambda(u,v)$, $u \in U$ \nllabel{alg:gh:minCut}\\
		$S_u \gets S \cap U$ and $S_v \gets S \cap (V_S \setminus U)$ \myco*[f]{split $S = S_u \dcup S_v$}\\
		$V_* \gets (V_* \setminus \{S\}) \cup \{S_u,S_v\}$, $E_* \gets E_* \cup \{\{S_u,S_v\}\}$, $c_*(S_u,S_v) \gets \lambda(u,v)$\\
		\ForAll{former edges $e_j = \{S,S_j\} \in E_*$
			\nllabel{alg:gh:reconnectStart}}{
			\lIf({\myco*[f]{reconnect $S_j$ to $S_u$}}){$[N_j] \in U$}{
				$e_{j} \gets \{S_u,S_j\}$
			}
			\lElse({\myco*[f]{reconnect $S_j$ to $S_v$}}){
				$e_j \gets \{S_v,S_j\}$ \nllabel{alg:gh:reconnectEnd}
			}
		}
	}
	\Return $T_*$
\end{algorithm2e}

The \emph{intermediate} cut tree $T_* = (V_*,E_*,c_*)$
is initialized as an isolated, edgeless node containing all original vertices. 
Then, until each node of~$T_*$ is a singleton node, a node $S \in V_*$ is \emph{split}.
\begin{figure}[tb]
\centering
	\subfigure[If $x \in S_{u}$, $\{x,y\}$ is still a cut pair of $\{S_{u}, S_j\}$]{
		\label{fig:cutPairs_a}
		\includegraphics[width = 7cm, page=1]{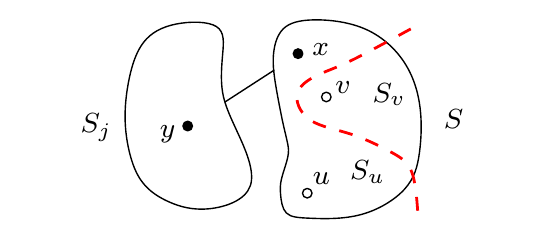}
	}
\hspace{1ex}
	\subfigure[If $x\notin S_{u}$, $\{u,y\}$ is a cut pair of  $\{S_u,S_j\}$]{
		\label{fig:cutPairs_b}
		\includegraphics[width = 7cm, page=2]{cutPairs.pdf}
	}
	\caption{Situation in Lemma~\ref{lem:cut_pairs}. There always exists a cut pair of the edge $\{S_{u}, S_j\}$ in the nodes incident to the edge, independent of the shape of the split cut (dashed).}
	\label{fig:cutPairs}
\end{figure}
To this end, nodes $S' \neq S$ are dealt with by contracting in $G$ whole subtrees $N_j$ of $S$ in $T_*$, connected to $S$ via edges $\{S,S_j\}$, to single nodes $[N_j]$
before cutting, which yields $G_S$.
The split of $S$ into $S_u$ and $S_v$ is then defined by a minimum $u$-$v$-cut (\emph{split cut}) in $G_S$,
which does not cross any of the previously used cuts due to the contraction technique. Afterwards, each $N_j$ is reconnected, again by $S_j$,
to either $S_u$ or $S_v$ depending on which side of the cut $[N_j]$ ended up.
Note that this cut in $G_S$ can be proven to induce a minimum $u$-$v$-cut in $G$.

The correctness of \textsc{Cut Tree} is guaranteed by Lemma~\ref{lem:cut_pairs}, 
which takes care for the \emph{cut pairs} of the reconnected edges.
It states that each edge $\{S,S'\}$ in~$T_*$ has a cut pair $\{x,y\}$ with $x\in S$, $y\in S'$. An intermediate cut tree satisfying this condition is \emph{valid}. The assertion is not obvious, since the nodes incident to the edges in~$T_*$ change whenever the edges are reconnected. Nevertheless, each edge in the final cut tree represents a minimum separating cut of its incident vertices, due to Lemma~\ref{lem:cut_pairs}.
The lemma was formulated and proven in~\cite{gh-mtnf-61} and rephrased in~\cite{g-vsmap-90}. See Figure~\ref{fig:cutPairs}.
\begin{lemma}[Gus.~\cite{g-vsmap-90}, Lem.~4] 
\label{lem:cut_pairs} 
	Let $\{S,S_j\}$ be an edge in $T_*$ inducing a cut with cut pair $\{x,y\}$, w.l.o.g.\ $x\in S$.
	Consider step pair $\{u,v\} \subseteq S$ that splits $S$ into $S_u$ and $S_v$, w.l.o.g.\ $S_j$ and $S_u$ ending up on the same cut side, i.e.\ $\{S_{u}, S_j\}$ becomes a new edge in $T_*$.
	If $x\in S_{u}$, $\{x,y\}$ remains a cut pair for $\{S_{u}, S_j\}$.
	If $x\in S_{v}$, $\{u,y\}$ is also a cut pair of $\{S_{u}, S_j\}$.
\end{lemma}
While Gomory and Hu use contractions in $G$ to prevent crossings of the cuts, as a simplification, Gusfield introduced the following lemma showing that contractions are not necessary, since any arbitrary minimum separating cut can be bent along the previous cuts resolving any potential crossings. See~Figure~\ref{fig:shelteredByPrev}.
\begin{lemma}[Gus.~\cite{g-vsmap-90}, Lem.~1]\label{lem:shelteredByPrev}
	Let $(X,V \setminus X)$ be a minimum $x$-$y$-cut in $G$, with $x \in X$.
	Let $(H, V \setminus H)$ be a minimum $u$-$v$-cut, with $u,v \in V \setminus X$ and $x \in H$.
	Then the cut $(H \cup X, (V \setminus H)\cap(V \setminus X))$ is also a minimum $u$-$v$-cut.
\end{lemma}
\begin{figure}[t]
\centering
	\includegraphics[width = 6cm]{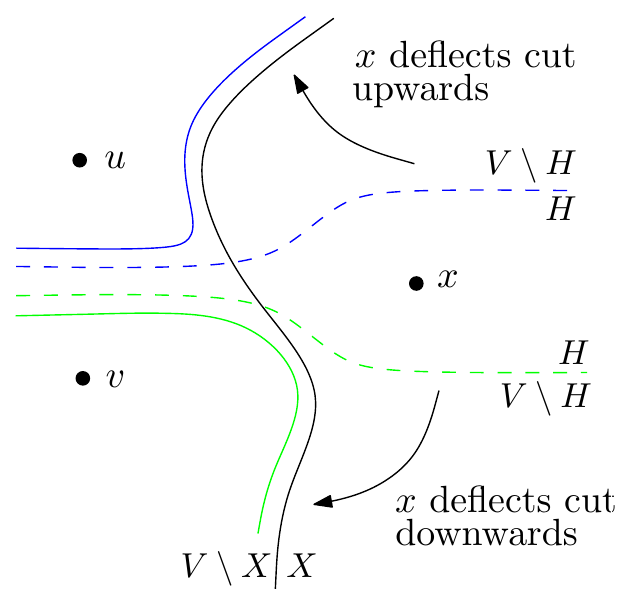}
	\caption{Depending on $x$ Lem.~\ref{lem:shelteredByPrev} bends the cut $(H, V \setminus H)$ upwards or downwards.}
	\label{fig:shelteredByPrev}
\end{figure}

%

%
We say that $(X,V\setminus X)$ \emph{shelters} $X$, meaning that each minimum $u$-$v$-cut with $u,v \notin X$ can be reshaped, such that it does no longer split $X$.
%
\paragraph{Representation of Intermediate Trees.}
In the remainder of this work we represent each node in $T_*$, which consists of original vertices in $G$, by an arbitrary tree of \emph{thin} edges connecting the contained vertices in order to indicate their membership to the node. An edge connecting two nodes in $T_*$ is represented by a \emph{fat} edge, which we connect to an arbitrary vertex in each incident compound node. 
Fat edges represent minimum separating cuts in $G$.
If a node contains only one vertex, we color this vertex black. Black vertices are only incident to fat edges. The vertices in non-singleton nodes are colored white. White  vertices are incident to at least one thin edge.
In this way, $T_*$ becomes a tree on $V$ with two types of edges and vertices. For an example see Figure~\ref{fig:stratingPoints}.
%
\paragraph{Conditions for Reusing Cuts.}
Consider a set $K$ of $k \leq n-1$ cuts in $G$ for example given by a previous cut tree in a dynamic scenario. The following
theorem states sufficient conditions for $K$, such that there exists a valid intermediate cut tree that represents exactly the cuts in $K$. Such a tree can then be further processed to a proper tree by \textsc{Cut Tree}, saving at least $|K|$ cut computations compared to a construction from scratch.
\begin{theorem}
\label{the:tool}
 Let $K$ denote a set of non-crossing minimum separating cuts in $G$ and let~$F$ denote a set of associated cut pairs such that each cut in $K$ separates exactly one pair in~$F$. Then there exists a valid intermediate cut tree representing exactly the cuts in~$K$.
\end{theorem}
\begin{proof}
Theorem~\ref{the:tool} follows inductively from the correctness of \algo{Cut Tree}.
Consider a run of \textsc{Cut Tree} that uses the elements in $F$ as step pairs in an arbitrary order and the associated cuts in $K$ as split cuts.
Since the cuts in $K$ are non-crossing each separating exactly one cut pair in $F$, splitting a node neither causes reconnections nor the separation of a pair that was not yet considered.  Thus, \textsc{Cut Tree} reaches an intermediate tree representing the cuts in $K$ with the cut pairs located in the incident nodes.
\end{proof}
With the help of Theorem~\ref{the:tool} we can now construct a valid intermediate cut tree from the cuts that remain valid after a change of $G$ according to Lemma~\ref{lem:path}. 
These cuts are non-crossing as they are represented by tree edges, and the vertices incident to these edges constitute a set of cut pairs as required by Theorem~\ref{the:tool}.  
The resulting tree for an inserted edge or an increased edge cost is shown in Figure~\ref{fig:point_add}. In this case, all but the edges on $\pi(b,d)$ can be reused. Hence, we draw these edges fat. The remaining edges are thinly drawn.
The vertices are colored according to the compound nodes indicated by the thickness of the edges. 
Vertices incident to a fat edge correspond to a cut pair.

For a deleted edge or a decreased edge cost, the edges on $\pi(b,d)$ are fat, while the edges that do not lie on $\pi(b,d)$ are thin (cp.~Figure~\ref{fig:point_del}). Furthermore, the costs of the fat edges decrease by $\Delta$, since they all cross the changing edge $\{b,d\}$ in $G$.
Compared to a construction from scratch, starting the \textsc{Cut Tree} routine from these intermediate trees already saves $n-1-|\pi(b,d)|$ cut computations in the first case
and $|\pi(b,d)|$ cut computations in the second case, where $|\pi(b,d)|$ counts the edges on $\pi(b,d)$.
Hence, in scenarios with only little varying path lengths and a balanced number of increasing and decreasing costs, we can already save about half of the cut computations.
We further remark that the result of Barth et.~al.~\cite{bbdf-rm-06}, who costly prove the existence of the intermediate cut tree in Figure~\ref{fig:point_add}, easily follows by Theorem~\ref{the:tool} applied to the cuts in Lemma~\ref{lem:path} as seen above.
In the following we want to use even more information from the previous cut tree $T(G)$ when executing \textsc{Cut Tree} unfolding the intermediate tree to a proper cut tree of $(n-1)$ fat edges.
\begin{figure}[t]
\centering
	\subfigure[Intermediate cut tree for $\Gp$.]{
		\label{fig:point_add}
		\includegraphics[width = 7cm, page=2]{firstTree.pdf}
	}
\hspace{2ex}
	\subfigure[Intermediate cut tree for $\Gm$.]{
		\label{fig:point_del}
		\includegraphics[width = 7cm, page=1]{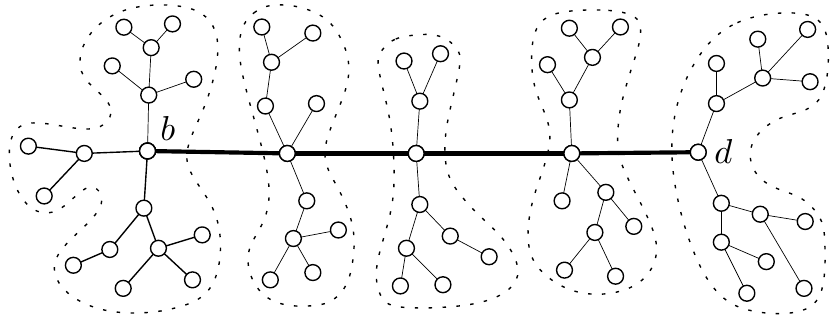}
	}
	\caption{Intermediate cut trees in dynamic scenarios. Fat edges represent valid minimum cuts, thin edges indicate compound nodes. Contracting the thin edges yields nodes of white vertices (indicated by dotted lines). Black vertices correspond to singletons.}
	\label{fig:stratingPoints}
\end{figure}
The next section lists further lemmas that allow the reuse of cuts already given by~$T(G)$.
%
\paragraph{Further Reusable Cuts.} 
In this section we focus on the reuse of those cuts that are still represented by thin edges in Figure~\ref{fig:stratingPoints}. 
If $\{b,d\}$ is inserted or the cost increases, the following corollary obviously holds, since $\{b,d\}$ crosses each minimum $b$-$d$-cut.
%
%
\newcommand{\lemmaReuseEdgeIns}{If $\{b,d\}$ is newly inserted with 
$c^\oplus(b,d) = \Delta$ or $c(b,d)$ increases by $\Delta$, any minimum $b$-$d$-cut in $G$ 
remains valid in $\Gp$ with $\lambda^\oplus(b,d) = \lambda(b,d) + \Delta$.}
\begin{corollary}
\label{cor:reuseEdge_Ins}
	\lemmaReuseEdgeIns
\end{corollary}
\noindent
Note that reusing a valid minimum $b$-$d$-cut as split cut in \textsc{Cut Tree} separates $b$ and $d$ such that $\{b,d\}$ cannot be used again as step pair in a later iteration step. This is, we can reuse only one minimum $b$-$d$-cut, even if there are several such cuts represented in~$T(G)$.
Together with the following corollary, Corollary~\ref{lem:reuseEdge_Ins} directly allows the reuse of the whole cut tree~$T(G)$ if $\{b,d\}$ is an existing bridge in $G$ (with increasing cost). 
%
%
\newcommand{\corollaryBridgeDetec}{An edge $\{u,v\}$ is a bridge in $G$ iff $ c(u,v) = \lambda(u,v) > 0$. Then $\{u,v\}$ is also an edge in $T(G)$ representing the cut that is given by the two sides of the bridge.}
\begin{corollary}
\label{cor:bridgeDetec}
\corollaryBridgeDetec
\end{corollary}
%
\noindent
While the first part of Corollary~\ref{cor:bridgeDetec} is obvious, the second part follows by the fact that a bridge induces a minimum separating cut for all vertices on different bridge sides, while it does not cross any minimum separating cut for vertices on a common side.
If $G$ is disconnected and $\{b,d\}$ is a new bridge in $\Gp$, reusing the whole tree is also possible by replacing a single edge. Such bridges can be easily detected having the cut tree $T(G)$ at hand, since $\{b,d\}$ is a new bridge if and only if $\lambda(b,d) = 0$. New bridges particularly occur if newly inserted vertices are connected for the first time.
\newcommand{\lemmaReuseTreeIns}{
	Let $\{b,d\}$ be a new bridge in $\Gp$. 
	Then replacing an edge of cost~$0$ by~$\{b,d\}$ with cost $c^\oplus(b,d)$ 
	on $\pi(b,d)$ in $T(G)$ yields a new cut tree $T(\Gp)$. 
}
\begin{lemma}
\label{lem:reuseTree_Ins}
	\lemmaReuseTreeIns
\end{lemma}
\begin{proof}
	Since $\{b,d\}$ is a new bridge, $b$ and $d$ are in different connected components in $G$. 
	Hence, swapping one of these components on the other side of a cut that previously separated $b$ and $d$ such that both 
	components are on a common side does not change the cost of the cut,
	as the set of edges crossing the cut in $G$ remains the same.
	The edge of cost $0$ in $T(G)$ that is replaced by $\{b,d\}$ lies on $\pi(b,d)$ and thus deleting 
	this edge yields two connected components in $T(G)$ that correspond to the connected components 
	containing $b$ and $d$ in $G$. Reconnecting these components by $\{b,d\}$ in $T(G)$ 
	yields again a tree and equals the swapping of one component to the side of the other 
	component for each cut in $T(G)$ that previously separated $b$ and $d$. All other cuts in $T(G)$ remain the same. 
	Hence, after the replacement, the remaining edges in the resulting tree still represent minimum 
	separating cuts with respect to the same cut pairs as before, while the new edge $\{b,d\}$ obviously 
	represents a minimum $b$-$d$-cut in $\Gp$.
\end{proof}
\noindent
If $\{b,d\}$ is deleted or the cost decreases, handling bridges (always detectable by Corollary~\ref{cor:bridgeDetec}) is also easy.
\newcommand{\lemmaReuseTreeDel}{
	If $\{b,d\}$ is a bridge in $G$ and the cost 
	decreases by $\Delta$ (or $\{b,d\}$ is deleted), 
	decreasing the edge cost on 
	$\pi(b,d)$ in $T(G)$ by $\Delta$ yields a 
	new cut tree $T(\Gm)$. 
}
\begin{lemma}
\label{lem:reuseTree_Del}
	\lemmaReuseTreeDel
\end{lemma}
\begin{proof}
	According to Corollary~\ref{cor:bridgeDetec}, it is $\pi(b,d) = \{b,d\}$ in 
	$T(G)$, and $\{b,d\}$ remains a valid cut with cost $\lambda(b,d) - \Delta$ 
	in $\Gm$, by Lemma~\ref{lem:path}.
	All other edges in $T(G)$ represent minimum separating cuts in $G$ 
	with respect to vertices that lie on a common cut side. In particular these 
	cuts do not separate $b$ and $d$.
	Hence, any cheaper cut in $\Gm$ would also not separate $b$ and $d$, 
	and thus, would have been already cheaper in $G$. 
\end{proof}

\noindent
If $\{b,d\}$ is no bridge, at least other bridges in $G$ can still be reused if $\{b,d\}$ is deleted or the edge cost decreases.
Observe that a minimum separating cut in~$G$ only becomes invalid in $\Gm$ if there is a cheaper cut in $\Gm$ that separates the same vertex pair. Such a cut necessarily crosses the changing edge $\{b,d\}$ in $G$, since otherwise it would have been already cheaper in~$G$.
Hence, an edge in $E_T$ corresponding to a bridge in $G$ cannot become invalid, since any cut in~$\Gm$ that crosses $\{b,d\}$ besides the bridge would be more expensive.
In particular, this also holds for zero-weighted edges in $E_T$.
\begin{corollary}\label{cor:reuseOtherBridge}
Let $\{u,v\}$ denote an edge in $T(G)$ with $c_T(u,v) = 0$ or an edge that corresponds to a bridge in $G$. Then $\{u,v\}$ is still a minimum $u$-$v$-cut in $\Gm$.
\end{corollary}
Lemma~\ref{lem:unfoldDel} shows how a cut 
that is still valid in $\Gm$ may allow the reuse of all edges in~$E_T$ that lie on one cut side. Figure~\ref{fig:unfold} shows an example.
Lemma~\ref{lem:reuseByCosts} says that a cut 
that is cheap enough, cannot become invalid in $\Gm$. 
Note that the bound considered in this context depends on the current intermediate tree.
\newcommand{\lemmaUnfoldDel}{
	Let $(U, V \setminus U)$ be a minimum 
	$u$-$v$-cut in $\Gm$ with $\{b,d\} \subseteq V\setminus U$ 
	and $\{g,h\} \in E_T$ with $g,h \in U$.
	Then $\{g,h\}$ is a minimum separating cut in $\Gm$ 
	for all its previous cut pairs within $U$.
}
\begin{lemma}
\label{lem:unfoldDel}
	\lemmaUnfoldDel
\end{lemma}
\begin{proof}
	Suppose there exists a minimum $g$-$h$-cut in $\Gm$ that is cheaper 
	than the cut represented by $\{g,h\}$. Note that the cut $\{g,h\}$ costs 
	the same in $G$ and $\Gm$.
	Such a cheaper minimum $g$-$h$-cut in $\Gm$ would separate $b$ 
	and $d$ in $V\setminus U$. At the same time, Lemma~\ref{lem:shelteredByPrev} 
	would allow to bend such a cut along $V\setminus U$ such that it induces a 
	minimum $g$-$h$-cut that does not separate $b$ and $d$. 
	The latter would have been already cheaper in $G$.
\end{proof}
%
\newcommand{\lemmaReuseByCosts}{
	Let $T_* = (V, E_*, c_*)$ denote a valid intermediate cut tree for $\Gm$, 
	where all edges on $\pi(b,d)$ are fat and let $\{u,v\}$ be a thin edge with 
	$v$ on $\pi(b,d)$ such that $\{u,v\}$ represents a minimum $u$-$v$-cut in $G$. 
	Let $N_{\pi}$ denote the set of neighbors of $v$ on $\pi(b,d)$. 
	If $\lambda(u,v) \leq \min_{x \in N_{\piâ}}\{c_*(x,v)\}$, then $\{u,v\}$ 
	is a minimum $u$-$v$-cut in $\Gm$.
}
\begin{lemma}
\label{lem:reuseByCosts}
	\lemmaReuseByCosts
\end{lemma}
\begin{proof}
	The edges on $\pi(b,d)$ incident to $v$ already represent minimum 
	separating cuts in~$\Gm$. Any new $u$-$v$-cut in $\Gm$ that is cheaper 
	than the cut represented by $\{u,v\}$, must separate $b$ and $d$, i.e., must separate 
	two adjacent vertices on $\pi(b,d)$. However, the fat edges incident to $v$ on 
	$\pi(b,d)$ shelter the remaining path edges from being separated by a new 
	cut (cp.~Lemma~\ref{lem:shelteredByPrev}). Thus, there is a new cut that separates 
	$v$ from exactly one of its neighbors on the path, denoted by~$x$. This is, the new 
	cut must be at least as expensive as the cost of a minimum $x$-$v$-cut  in $\Gm$, 
	which is not possible if $\lambda(u,v)$ in $G$ is already at most equal.
\end{proof}
\begin{figure}[tb]
\centering
	\subfigure[Edges in $U$ remain valid, cp.~Lemma~\ref{lem:unfoldDel}.]{
		\label{fig:unfold}
		\includegraphics[width = 7cm, page=1]{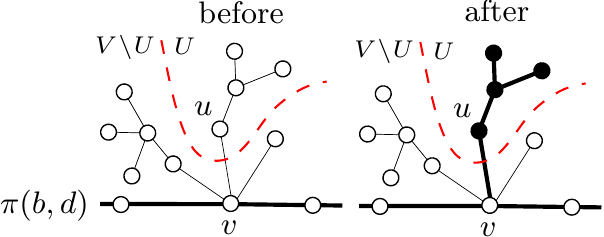}
	}
	\subfigure[Reshaping new cut by reconnecting edges.]{
		\label{fig:delalgol}
		\includegraphics[width = 7.2cm, page=1]{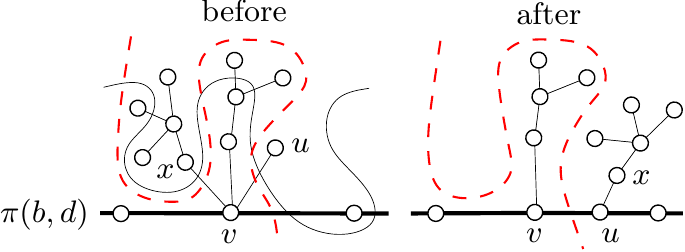}
	}
	\caption{(a) cut $\{u,v\}$ remains valid, subtree $U$ can be reused. (b) new cheaper cut for $\{u,v\}$ (black) can be reshaped by Theo.~\ref{theo:delSelteredByOld}, Lem.~\ref{lem:shelteredByPrev} (dashed), $\{u,v\}$ becomes a fat edge.}
	\label{fig:del_algo}
\end{figure}
%
\section{The Dynamic Cut Tree Algorithm}\label{sec:DynCutTreeAlgo}
In this section we introduce one update routine for each type of change: inserting a vertex, deleting a vertex, increasing an edge cost or inserting an edge, decreasing an edge cost or deleting an edge. These routines base on the static iterative approach but involve the lemmas from Sec.~\ref{sec:Theo} in order to save cut computations. We again represent intermediate cut trees by fat and thin edges, which simplifies the reshaping of cuts.

We start with the routines for vertex insertion and deletion, which trivially abandon cut computations. We leave the rather basic proofs of correctness to the reader. 
A vertex~$b$ inserted into $G$ forms a connected component in $\Gp$. Hence, we insert $b$ into $T(G)$ connecting it to the remaining tree by an arbitrary zero-weighted edge. If $b$ is deleted from $G$, it was a single connected component in $G$ before. Hence, in $T(G)$ $b$ is only incident to zero-weighted edges. Deleting $b$ from $T(G)$ and reconnecting the resulting subtrees by arbitrary edges of cost $0$ yields a valid intermediate cut tree for $\Gm$. 

The routine for increasing an edge cost or inserting an edge first checks if $\{b,d\}$ is a (maybe newly inserted) bridge in $G$. In this case, it adapts $c_T(b,d)$ according to Corollary~\ref{lem:reuseEdge_Ins} if $\{b,d\}$ already exists in $G$, and rebuilds $T(G)$ according to Lemma~\ref{lem:reuseTree_Ins} otherwise. Both requires no cut computation. 
If $\{b,d\}$ is no bridge, the routine constructs the intermediate cut tree shown in Figure~\ref{fig:point_add}, reusing all edges that are not on $\pi(b,d)$. Furthermore, it chooses one edge on $\pi(b,d)$ that represents a minimum $b$-$d$-cut in $\Gp$ and draws this edge fat (cp.~Corollary~\ref{lem:reuseEdge_Ins}). The resulting tree is then further processed by \textsc{Cut Tree}, which costs  $|\pi(b,d)|-1$ cut computations and is correct by Theorem~\ref{the:tool}.
%
\begin{algorithm2e}[h]
\caption{\textsc{Decrease or Delete}}\label{alg:edgeDel}
\SetKwComment{tco}{\%}{}
\KwIn{$T(G)$, $b,d$, $c(b,d)$, $c^\ominus(b,d)$, $\Delta := c(b,d) - c^\ominus(b,d)$}
\KwOut{$T(\Gm)$}
$T_* \gets T(G)$\\
\lIf{\emph{$\{b,d\}$ is a bridge}} {apply Lemma~\ref{lem:reuseTree_Del}; \Return $T(\Gm) \gets T_*$}\nllabel{ln:bridge} \\
Construct intermediate tree according to Figure~\ref{fig:point_del} \nllabel{ln:tree}\\
$Q \gets$ thin edges non-increasingly ordered by their costs \nllabel{ln:sort} \\
 \While{$Q \not= \emptyset$}{ 
 	$\{u,v\} \gets $ most expensive thin edge with $v$ on $\pi(b,d)$\nllabel{ln:steppairs}\\
	$N_{\pi} \gets$ neighbors of $v$ on $\pi(b,d)$; $L \gets \min_{x\in N_{\pi}}\{c_*(x,v)\}$\nllabel{ln:L}\\
	\If( \myco*[f]{Lem.~\ref{lem:reuseByCosts} and Cor.~\ref{cor:reuseOtherBridge}}){ $L\geq \lambda(u,v)$ or $\{u,v\} \in E$ with $\lambda(u,v) = c(u,v)$ }{\nllabel{ln:check}
		draw $\{u,v\}$ as a fat edge\\
		consider the subtree $U$ rooted at $u$ with $v\notin U$, \myco*[f]{Lem.~\ref{lem:unfoldDel} and Fig.~\ref{fig:unfold}} \nllabel{ln:ok} \\
		draw all edges in $U$ fat, remove fat edges from $Q$\\
		continue loop
	}
  	$(U,V\setminus U) \gets $ minimum $u$-$v$-cut in $\Gm$ with $u \in U$ \nllabel{ln:cut}\\
	draw $\{u,v\}$ as a fat edge, remove $\{u,v\}$ from $Q$\\
	\lIf{ $\lambda(u,v) = c^\ominus(U,V\setminus U)$ }{\nllabel{ln:recomp}
		goto line \ref{ln:ok} \myco*[f]{old cut still valid}\\
	}
	$c_*(u,v) \gets c^\ominus(U,V\setminus U)$ \myco*[f]{otherwise}\\
	$N \gets$ neighbors of $v$ \nllabel{ln:N}\\
	\ForAll(\myco*[f]{bend split cut by~Theo.~\ref{theo:delSelteredByOld} and Lem.~\ref{lem:shelteredByPrev}}){$x\in N$ }{
		\lIf{ $x\in U$ }{
			reconnect $x$ to $u$ \nllabel{ln:for}\\
		}
	}
 }
\Return $T(\Gm) \gets T_*$
\end{algorithm2e}

The routine for decreasing an edge cost or deleting an edge is given by Algorithm~\ref{alg:edgeDel}. 
We assume $G$ and $\Gm$ to be available as global variables.
Whenever the intermediate tree $T_*$ changes during the run of Algorithm~\ref{alg:edgeDel}, the path $\pi(b,d)$ is implicitly updated without further notice. 
Thin edges are weighted by the old connectivity, fat edges by the new connectivity of their incident vertices. 
Whenever a vertex is reconnected, the newly occurring edge inherits the cost and the thickness from the disappearing edge.

Algorithm~\ref{alg:edgeDel} starts by checking if $\{b,d\}$ is a bridge (line~\ref{ln:bridge}) and reuses the whole cut tree $T(G)$ with adapted cost $c_T(b,d)$ (cp.~Lemma~\ref{lem:reuseTree_Del}) in this case.
Otherwise (line~\ref{ln:tree}), it constructs the intermediate tree shown in Figure~\ref{fig:point_del}, reusing all edges on $\pi(b,d)$ with adapted costs.
Then it proceeds with iterative steps similar to \textsc{Cut Tree}. However, the difference is, that the step pairs are not chosen arbitrarily, but according to the edges in $T(G)$, starting with those edges that are incident to a vertex $v$ on $\pi(b,d)$ (line~\ref{ln:steppairs}). In this way, each edge $\{u,v\}$ which is found to remain valid in line~\ref{ln:check} 
or line~\ref{ln:recomp} 
allows to retain a maximal subtree
(cp.~Lemma~\ref{lem:unfoldDel}), since $\{u,v\}$
is as close as possible to $\pi(b,d)$. The problem however is that cuts that are no longer valid, must be replaced by new cuts, which not necessarily respect the tree structure of $T(G)$. This is, a new cut possibly separates adjacent vertices in $T(G)$, which hence cannot be used as a step pair in a later step. Thus, we potentially miss valid cuts and the chance to retain further subtrees.

We solve this problem by reshaping the new cuts in the spirit of Gusfield.
Theorem~\ref{theo:delSelteredByOld} shows how arbitrary cuts in $\Gm$ (that separate $b$ and $d$) can be bend along old minimum separating cuts in $G$ without becoming more expensive (see Figure~\ref{fig:bendAlongOld}). 
\newcommand{\theoDelSelteredByOld}{
	Let $(X,V\setminus X)$ denote a minimum $x$-$y$-cut in $G$ with 
	$x \in X$, $y\in V\setminus X$ and $\{b,d\}\subseteq V\setminus X$. 
	Let further $(U, V\setminus U)$ denote a cut that separates $b,d$. 
	If (i)  $(U, V\setminus U)$ separates $x, y$ with $x\in U$, then 
	$c^\ominus(U\cup X, V\setminus (U\cup X)) \leq c^\ominus(U, V\setminus U)$. 
	If (ii) $(U, V\setminus U)$ does not separate $x, y$ with $x\in V\setminus U$, 
	then $c^\ominus(U\setminus X, V\setminus (U\setminus X)) \leq c^\ominus(U, V\setminus U)$.
}
\begin{theorem}
\label{theo:delSelteredByOld}
	\theoDelSelteredByOld
\end{theorem}
\begin{figure}[t]
\centering
	\subfigure[Deflected by $x$, Theorem~\ref{theo:delSelteredByOld}(i) bends $(U,V\setminus U)$ downwards along $X$.]{
		\label{fig:cutsInTheta_a}
		\includegraphics[width = 6.5cm]{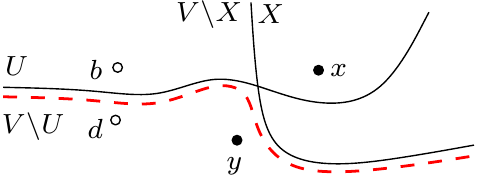}
	}
\hspace{4ex}
	\subfigure[Deflected by $x$, Theorem~\ref{theo:delSelteredByOld}(ii) bends $(U,V\setminus U)$ upwards along $X$.]{
		\label{fig:cutsInTheta_b}
		\includegraphics[width = 6.5cm]{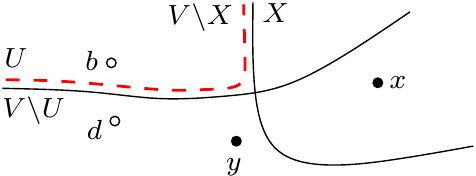}
	}
	\caption{Situation of Theorem~\ref{theo:delSelteredByOld}. Reshaping cuts in $\Gm$ along previous cuts in $G$.}
	\label{fig:bendAlongOld}
\end{figure}
\begin{proof}
	We prove Theorem~\ref{theo:delSelteredByOld}(i) by contradiction, using the fact that $(X,V\setminus X)$ is a minimum $x$-$v$-cut in $G$.
		We show that $(U \cap X, V\setminus (U\cap X))$ 
	would have been cheaper than the minimum $x$-$v$-cut $(X,V\setminus X)$ in $G$ 
	if $c^\ominus(U, V\setminus U)$ was cheaper than $c^\ominus(U \cup X, V\setminus (U\cup X))$ in $G^\ominus$.
	We express the costs of $(X\cap U,V\setminus (X\cap U))$ and $(X,V\setminus X)$ with the help of $(U,  V\setminus U)$ and $(X \cup U, V\setminus (X\cup U))$ considered in Theorem~\ref{theo:delSelteredByOld}(i).
	Note that $(X\cap U,V\setminus (X\cap U))$ and $(X,V\setminus X)$ do not separate $b$ and $d$. Thus, their costs are unaffected by the deletion and it makes no difference whether we consider the costs in $G$ or $\Gm$. We get\medskip

\noindent\begin{tabularx}{\textwidth}{rlclclclcl}
	\textit{(i)} &$c(X\cap U,V\setminus (X\cap U))$  &=& $c^\ominus(U, V \setminus U)$ & &\\ &  &-& $c^\ominus(U\setminus X, V\setminus U)$ &+& $c^\ominus(U\setminus X,X\cap U)$\\
	\textit{(ii)} &$c(X,V\setminus X)$  &=& $c^\ominus(X \cup U, V\setminus (X\cup U))$ & & \\ & &-& $c^\ominus(U\setminus X,V\setminus (X\cup U))$ &+& $c^\ominus(U\setminus X,X)$	
	\end{tabularx}\medskip

	\noindent
	Since $V\setminus(X\cup U) \subseteq V\setminus U$, it is $c^\ominus(U\setminus X,V\setminus(X\cup U)) \leq  c^\ominus(U\setminus X,V\setminus U)$. From
	$X\cap U \subseteq X$ further follows that $c^\ominus(U\setminus X,X\cap U) \leq c^\ominus(U\setminus X,X)$; together with the assumption that
	$\cmm(U, V\setminus U) < \cmm(X\cup U,V\setminus (X\cup U))$, we see the following if we subtract \textit{(ii)} from \textit{(i)}:
	\begin{eqnarray}
	c(X\cap U,V\setminus (X\cap U)) &-& c(X,V\setminus X)\nonumber \\ 
	& = & [c^\ominus(U, V\setminus U)- c^\ominus(X\cup U, V\setminus(X\cup U))]\nonumber\\
	& - & [c^\ominus(U\setminus X, V\setminus U) - c^\ominus(U\setminus X, V\setminus(X\cup U))]\nonumber\\
	& + & [c^\ominus(U\setminus X,X\cap U) - c^\ominus(U\setminus X,X)] < 0\nonumber
	\end{eqnarray}
	This contradicts the fact that $(X,V\setminus X)$ is a minimum $x$-$v$-cut in $G$.
%
%

	We prove Theorem~\ref{theo:delSelteredByOld}(ii) with the help of the same technique.
		We show that $(X\setminus U, V\setminus (X\setminus U))$ 
	would have been cheaper than the minimum $x$-$v$-cut $(X,V\setminus X)$ in $G$ 
	if $c^\ominus(U, V\setminus U)$ was cheaper than $c^\ominus(U \setminus X, V\setminus (U\setminus X))$ in $G^\ominus$.
	We express the costs of $(X\setminus U,V\setminus (X\setminus U))$ and $(X,V\setminus X)$ with the help of $(U,  V\setminus U)$ and $(U\setminus X, V\setminus (U\setminus X))$ considered in Theorem~\ref{theo:delSelteredByOld}(ii).
	Note that $(X\setminus U,V\setminus (X\setminus U))$ and $(X,V\setminus X)$ do not separate $b$ and $d$. Thus, their costs are unaffected by the deletion and it makes no difference whether we consider the costs in $G$ or $\Gm$. We get\medskip

\noindent\begin{tabularx}{\textwidth}{rlclclclcl}
	\textit{(i)} &$c(X\setminus U,V\setminus (X\setminus U))$  &=& $c^\ominus(U, V \setminus U)$ & &\\ &  &-& $c^\ominus(U, V\setminus (X\cup U))$ &+& $c^\ominus(X\setminus U,V\setminus (X\cup U))$\\
	\textit{(ii)} &$c(X,V\setminus X)$  &=& $c^\ominus(U \setminus X, V\setminus (U\setminus X))$ & & \\ & &-& $c^\ominus(U\setminus X,V\setminus (X\cup U))$ &+& $c^\ominus(X,V\setminus (X\cup U))$	
	\end{tabularx}\medskip

	\noindent
	Since $U\setminus X \subseteq U$, it is $c^\ominus(U\setminus X,V\setminus(X\cup U)) \leq  c^\ominus(U,V\setminus (X\cup U))$. From
	$X\setminus U \subseteq X$ further follows that $c^\ominus(X\setminus U,V\setminus (X\cup U)) \leq c^\ominus(X,V\setminus (X\cup U))$; together with the assumption that
	$\cmm(U, V\setminus U) < \cmm(U \setminus X, V\setminus (U\setminus X))$, we see the following if we subtract \textit{(ii)} from \textit{(i)}:
	
	\begin{eqnarray}
	c(X\setminus U,V\setminus (X\setminus U)) &-& c(X,V\setminus X)\nonumber \\ 
	& = & [c^\ominus(U, V\setminus U)- c^\ominus(U\setminus X, V\setminus(U\setminus X))]\nonumber\\
	& - & [c^\ominus(U, V\setminus (X\cup U)) - c^\ominus(U\setminus X, V\setminus(X\cup U))]\nonumber\\
	& + & [c^\ominus(X\setminus U,V\setminus (X\cup U)) - c^\ominus(X, V\setminus (X\cup U))] < 0\nonumber
	\end{eqnarray}
	This contradicts the fact that $(X,V\setminus X)$ is a minimum $x$-$v$-cut in $G$.
\end{proof}

\noindent
Since any new cheaper cut found in line~\ref{ln:cut} needs to separate $b$ and $d$, we can apply Theorem~\ref{theo:delSelteredByOld} to this cut regarding the old cuts that are induced by the other thin edges $\{x,v\}$ incident to $v$. As a result, the new cut gets reshaped without changing its cost such that each subtree rooted at a vertex $x$ is completely assigned to either side of the reshaped cut (line~\ref{ln:for}), depending on if the new cut separates $x$ and $v$  (cp. Figure~\ref{fig:delalgol}). Furthermore, Lemma~\ref{lem:shelteredByPrev} allows the reshaping of the new cut along the cuts induced by the fat edge on $\pi(b,d)$ that are incident to $v$. This ensures that the new cut does not cross parts of $T_*$ that are beyond these flanking fat edges. Since after the reshaping exact one vertex adjacent to~$v$ on $\pi(b,d)$ ends up on the same cut side as $u$, $u$ finally becomes a part of $\pi(b,d)$.

It remains to show that after the reconnection the reconnected edges are still incident to one of their cut pairs in $\Gm$ (for fat edges) and $G$ (for thin edges), respectively. For fat edges this holds according to Lemma~\ref{lem:cut_pairs}. For thin edges the order in line~\ref{ln:sort} guarantees that an edge $\{x,v\}$ that will be reconnected to $u$ in line~\ref{ln:for} is at most as expensive as the current edge $\{u,v\}$, and thus, also induces a minimum $u$-$x$-cut in $G$.
This allows applying Lemma~\ref{lem:unfoldDel} and~\ref{lem:reuseByCosts} as well as the comparison in line~\ref{ln:recomp} to reconnected thin edges, too.
Observe that an edge corresponding to a bridge never crosses a new cheaper cut, and thus, gets never reconnected.
In the end all edges in $T_*$ are fat, since each edge is either a part of a reused subtree or was considered in line~\ref{ln:steppairs}. Note that reconnecting a thin edge makes this edge incident to a vertex on $\pi(b,d)$ and decrements the hight of the related subtree.
%
\section{Performance of the Algorithm}\label{sec:Perform}
Unfortunately we cannot  give a meaningful guarantee on the number of saved cut computations. The saving depends on the length of the path $\pi(b,d)$, the number of $\{u,v\} \in E_T$ for which the connectivity $\lambda(u,v)$ changes, and the shape of the cut tree. 
In a star, for example, there exist no subtrees that could be reused by Lemma~\ref{lem:unfoldDel} (see Figure~\ref{fig:perform} (left) for a bad case example for edge deletion). 
Nevertheless, a first experimental proof of concept promises high practicability, particularly on graphs with less regular cut structures.
The instance we use is a network of e-mail communications within the Department of Informatics at KIT~\cite{-ddiki-11}. Vertices represent members, edges correspond to e-mail contacts, weighted by the number of e-mails sent between two individuals during the last 72 hours. We process a queue of 924\,900 elementary changes, which indicate the time steps in Figure~\ref{fig:perform} (right), and $923\,031$
of which concern edges.
 We start with an empty graph, constructing the network from scratch.
Figure~\ref{fig:perform} shows the ratio of cuts computed by the update algorithm and cuts needed by the static approach until the particular time step. The ratio is shown in total, and broken down to edge insertions (151\,169 occurrences), increasing costs (310473), edge deletions (151\,061) and decreasing costs (310\,328). The trend of the curves follows the evolution of the graph, which slightly densifies around time step $100\,000$ due to a spam-attack;
however, the update algorithm needs less than $4\%$ of the static computations even during this period.
We further observe that for decreasing costs, Theorem~\ref{theo:delSelteredByOld} together with Lemma~\ref{lem:shelteredByPrev} allows to contract all subtrees incident to the current vertex $v$ on $\pi(b,d)$, which shrinks the underlying graph to $deg_*(v)$ vertices, with $deg_*(v)$ the degree of $v$ in $T_*$. Such contractions could further speed up the single cut computations. Similar shrinkings can obviously be done for increasing costs, as well.
\begin{figure}[t]
\centering
\includegraphics[width = 15.5cm]{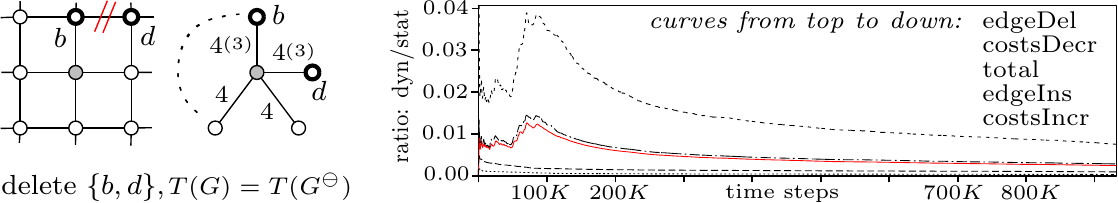}
\caption{left: $T(G)$ could be reused (new cost on $\pi(b,d)$ in brackets), but Alg.~\ref{alg:edgeDel} computes $n-3$ cuts. right: Cumulative ratio of dynamic and static cut computations.}
\label{fig:perform}
\end{figure}
%
\paragraph{Conclusion.}
We introduced a simple and  fast algorithm for dynamically updating a cut tree for a changing graph. In a first prove of concept our approach allowed to save over~$96\%$ of the cut computations and it provides even more possibilities for effort saving due to contractions. 
Currently, we are working on a more extensive experimental study, which we will present here as soon as we have finished. 
Recently, we further succeeded in improving the routine for an inserted edge or an increased cost such that it guarantees that each cut that remains valid is also represented by the new cut tree. This yields a high temporal smoothness, which is desirable in many applications. Note that the routine for a deleted edge or a decreased cost as presented in this work already provides this temporal smoothness.
%

\putbib[fastSimpleDynCutTree_tr]
 \end{bibunit}

\end{document}